\documentclass{article}

\usepackage{microtype}
\usepackage{graphicx}
\usepackage{subfigure}
\usepackage{booktabs} %
\usepackage{multirow}
\usepackage{amsmath, amsfonts}
\usepackage{pifont}%
\usepackage[frozencache]{minted}
\usepackage[bb=dsserif]{mathalpha}
\usepackage{siunitx}

\usepackage{hyperref}

\newcommand{\vertical}[1]{\rotatebox[origin=c]{90}{#1}}

\DeclareMathOperator{\ConvDown}{ConvDown}
\DeclareMathOperator{\Conformer}{Conformer}
\DeclareMathOperator{\ConvUp}{ConvUp}
\DeclareMathOperator{\Lin}{Linear}
\DeclareMathOperator{\LN}{LN}
\DeclareMathOperator{\MHA}{MHA}
\DeclareMathOperator{\Downsample}{Downsample}
\DeclareMathOperator{\MaskModule}{MaskModule}
\DeclareMathOperator{\FF}{FF}
\DeclareMathOperator{\MLP}{MLP}

\renewcommand\vec{\mathbf}
\newcommand{\ee}{\mathcal{E}}
\newcommand{\LL}{\mathcal{L}}
\newcommand{\EE}{E}
\newcommand{\seg}{\mathrm{seg}}
\DeclareMathOperator{\dur}{dur}

\newcommand{\cmark}{\ding{51}}
\newcommand{\xmark}{\ding{55}}

\usepackage[accepted]{icml2023}

\usepackage{amsmath}
\usepackage{amssymb}
\usepackage{mathtools}
\usepackage{amsthm}

\usepackage[capitalize,noabbrev]{cleveref}

\theoremstyle{plain}
\newtheorem{theorem}{Theorem}[section]
\newtheorem{proposition}[theorem]{Proposition}

\theoremstyle{definition}

\theoremstyle{remark}

\usepackage[textsize=tiny]{todonotes}

\icmltitlerunning{EEND-M2F: Masked-attention mask transformers for speaker diarization}

\begin{document}

\twocolumn[
\icmltitle{EEND-M2F: Masked-attention mask transformers for speaker diarization}

\icmlsetsymbol{equal}{*}

\begin{icmlauthorlist}
\icmlauthor{Marc H\"ark\"onen}{fano}
\icmlauthor{Samuel J. Broughton}{fano}
\icmlauthor{Lahiru Samarakoon}{fano}
\end{icmlauthorlist}

\icmlaffiliation{fano}{Fano Labs, Hong Kong}

\icmlcorrespondingauthor{Marc H\"ark\"onen}{marc.harkonen@fano.ai}

\icmlkeywords{Speech, Diarization, Transformers}

\vskip 0.3in
]

\printAffiliationsAndNotice{}  %

\begin{abstract}
In this paper, we make the explicit connection between image segmentation methods and end-to-end diarization methods.
From these insights, we propose a novel, fully end-to-end diarization model, EEND-M2F, based on the Mask2Former architecture.
Speaker representations are computed in parallel using a stack of transformer decoders, in which irrelevant frames are explicitly masked from the cross attention using predictions from previous layers.
EEND-M2F is lightweight, efficient, and truly end-to-end, as it does not require any additional diarization, speaker verification, or segmentation models to run, nor does it require running any clustering algorithms.
Our model achieves state-of-the-art performance on several public datasets, such as AMI, AliMeeting and RAMC.
Most notably our DER of 16.07\% on DIHARD-III is the first major improvement upon the challenge winning system.

\end{abstract}

\section{Introduction}\label{sec:intro}
\begin{table*}
    {\centering
    \begin{tabular}{llllllllll}
        \toprule
         & \multicolumn{9}{c}{Dataset}\\
        \cmidrule{2-10}
        {Model} & \vertical{Aishell-4} & \vertical{AliMeeting-far} & \vertical{AliMeeting-near} & \vertical{AMI-Mix} & \vertical{AMI-SDM} & \vertical{CALLHOME} & \vertical{DIHARD-III} & \vertical{RAMC} & \vertical{VoxConverse} \\
        \midrule
        VAD+VBx+OSD \cite{landini23diaper} & 15.84 & 28.84 & 22.59 & 22.42 & 34.61 & 26.18 & 20.28 & 18.19 & (6.12)\\
        EEND \cite{fujita19EEND-SA, horiguchi22eend-eda} & --- & --- & --- & 27.70 & --- & (21.19) & 22.64 & --- & --- \\
        EEND-EDA \cite{horiguchi22eend-eda} & --- & --- & --- & 15.80 & --- & (12.88) & 20.69 & --- & --- \\
        PyAnnote 3.1 \cite{plaquet23pyannote31} & 13.2 & 23.3 & --- & 18.0 & 22.9 & 28.4\textsuperscript{\textdagger} & 21.3 & 22.2\textsuperscript{\textdagger} & 10.4\\
        DiaPer \cite{landini23diaper} & 31.30 & 26.27 & 24.44 & 30.49 & 50.97 & 24.16 & 22.77 & 18.69 & (22.10) \\
        \midrule
        EEND-M2F (ours) & 15.56 & \textbf{13.20} & 10.77 & 13.86 & 19.83 & 21.28 & 16.28 & 11.13 & 15.99\\
        \quad +0.25s collar & (10.75) & (5.87) & (5.20) & (9.16) & (14.29) & (14.87) & (8.93) & (6.52) & (12.02)\\
        EEND-M2F + FT (ours) & 13.98 & 13.40 & \textbf{10.45} & \textbf{12.62} & \textbf{18.85} & 23.44 & \textbf{16.07} & \textbf{11.09} & 16.28\\
        \quad +0.25s collar & (9.34) & (6.11) & (5.02) & (7.92) & (13.33) & (16.72) & (8.82) & (6.46) & (12.36)\\
        \midrule
        State-of-the-art (as of Jan 2024) & \textbf{13.2} & 23.3 & 22.59 & 13.00 & 19.53 & (\textbf{10.08}) & 16.76 & 13.58 & (\textbf{4.0})\textsuperscript{*}\\
        Source & \cite{plaquet23pyannote31} &  \cite{plaquet23pyannote31} & \cite{landini23diaper} & \cite{chen23AED-EEND-EE} & \cite{he23ANSD-MA-MSE} & \cite{chen23AED-EEND-EE} & \cite{he23ANSD-MA-MSE} & \cite{samarakoon23EEND-TA} & \cite{baroudi23pyannote_vox}\\
        \bottomrule
    \end{tabular}}
    {
        \footnotesize
        {}\textsuperscript{\textdagger}From \url{https://github.com/pyannote/pyannote-audio/blob/develop/README.md}, as of commit 80634c9.\\
        {}\textsuperscript{*}Potentially biased, as model was tuned and validated on VoxConverse test set.
    }
    \caption{Diarization Error Rate (DER) across several datasets. Lower is better. Our single model EEND-M2F trained on all datasets simultaneously achieves state-of-the-art results, and performance is further improved by single dataset finetuning (+ FT). Error rates are computed with no collar, no oracle segments or speaker counts, and including overlapped speech. Values in parentheses are computed with a collar of 0.25s.}
    \label{tab:main}
\end{table*}

Image segmentation is the task of generating masks for sets of objects in an image.
Speaker diarization can be thought of as image segmentation for speech:
audio is viewed as a 1-dimensional image, and each speaker constitutes an object to be identified.
To give a binary mask for each speaker over time means effectively answering the question ``who spoke when''.

Despite this similarity between diarization and image segmentation, methods used in both communities have evolved in very different directions.
Traditional, cascaded diarization models consist of several components, such as voice activity detection, speaker embedding evaluation, and clustering \cite{park22dia_review}.
These components are trained separately on non-diarization objectives, and often do not handle overlapped speech.
In object detection, the preferred approach was R-CNN \cite{Girshick14r-cnn,Girshick15fast_r-cnn,ren15faster_r-cnn}, a two-stage approach combining region proposals and a CNN classifier.
Image segmentation is enabled by adding a mask prediction head \cite{He17mask_r-cnn}.
Alternatively, one can classify each pixel separately in a fully convolutional network \cite{long15fully}.

An initial convergence in methodology was achieved recently with the advent of the end-to-end neural diarization model (EEND) \cite{fujita19eend}, which directly outputs masks for a fixed number of speakers.
Its successor, the EEND-EDA \cite{horiguchi20eend-eda} model, adds a branch for computing speaker-wise embeddings, which allows it to handle a variable number of speakers.
The overall architecture of EEND-EDA has striking similarities to recent end-to-end segmentation models such as Maskformer \cite{cheng2021maskformer}.

End-to-end models remain dominant in the image space, with many recent state-of-the-art methods opting for Transformer \cite{vaswani17transformer} based models \cite{carion20DETR,zhu2020deformable,cheng2021maskformer,m2f,kirillov2023segment_anything,jain2023oneformer,Li23mask_dino}.
The development of diarization models however seems to have reverted to multi-stage, multi-model solutions, with many methods integrating EEND only as a subcomponent \cite{kinoshita21eend-vc,horiguchi21eend-gla,bredin23pyannote21}.
TS-VAD \cite{medennikov20TS-VAD} based models require auxiliary inputs, such as speaker embeddings or initial diarization predictions, to output refined diarization results.
In general, these approaches have outperformed recent end-to-end models \cite{rybicka22end,fujita23intermediate,landini23diaper}, while simultaneously requiring substantially more computational resources.

One of the goals with this work is to bridge the gap between modern image segmentation and diarization models, and reestablish end-to-end diarization models as viable, conceptually simple, and computationally efficient alternatives to current state-of-the-art systems.
To this end, we introduce a novel diarization model, coined end-to-end neural diarization with masked-attention mask transformers (EEND-M2F).
The model architecture is heavily influenced by that of Mask2Former \cite{m2f}, and we make no claim of originality in its design.
Rather, our aim is to show the effectiveness of modern advances in computer vision when applied to speech, hopefully encouraging further interaction between these two communities.
Our model achieves state-of-the-art performance on several datasets, as can be seen in \cref{tab:main}.

In \cref{sec:background}, we give a brief overview of recent advances in image segmentation and diarization, making explicit connections between the two fields.
\cref{sec:model} introduces EEND-M2F and motivates some of the design choices.
Results and ablation studies can be found in \cref{sec:experiments}.
In \cref{sec:discussion} we discuss possible future directions.
Finally, in \cref{app:der}, we present a simple and concise implementation of a computation routine for Diarization Error Rate (DER), the \emph{de facto} standard metric for diarization.
Since our implementation uses directly the same output masks and labels as in training, this version can be easily included in a validation or testing routine, with no conversions to segments or RTTM files needed.
By using masks instead of a list of segments, our implementation runs in parallel on a GPU, and provably outputs the DER corresponding to the optimal speaker permutation.
Our hope with this is to nudge the community into solutions that are clean, concise, and parallelizable.

To summarize, we list our main contributions:
\begin{itemize}
    \item We demonstrate that, with the right modifications, image segmentation models can be used almost as-is in speaker diarization.
    We motivate design choices in the context of diarization.
    \item We show that a simple, one-stage end-to-end model can still be competitive and outperform multi-component systems, model ensembles, and large self-supervised pretrained models.
    \item We achieve state-of-the-art results in several datasets, such as DIHARD-III (16.07\%), AMI-Mix (12.62\%), and AliMeeting (13.20\%), beating model ensembles, multi-stage models, and large pretrained models.
    Many of the state-of-the-art results can be achieved even without additional dataset-specific finetuning.
    Our model is evaluated using the most stringent settings, i.e.~with overlaps, no oracle segments or speaker counts, and no oracle voice activity detection (VAD).
\end{itemize}

\section{Related works}\label{sec:background}
\subsection{Image segmentation}
Common tasks in image segmentation include semantic segmentation, instance segmentation, and panoptic segmentation.
Borrowing the terminology of \cite{Kirillov19panoptic}, semantic segmentation can be though of the study of \emph{stuff}, amorphous regions of similar materials and textures.
An example output in semantic segmentation could be a set of masks delimiting the sky, grass, crowds, or buildings.
Instance segmentation on the other hand studies \emph{things}, often countable and enumerable.
For example, the output could be one mask for each person, one for each car, one for each dog.
Panoptic segmentation unifies both tasks, such that each pixel is classified into a class and instance.

A closely related task is object detection, which is the task of generating and classifying bounding boxes around objects of interest.
In fact, early neural segmentation models were often built on top of object detection models, one of the most notable being Mask-R-CNN \cite{He17mask_r-cnn}, which adds a mask prediction head to Faster-R-CNN \cite{ren15faster_r-cnn}.
Its two-stage approach first predicts bounding box proposals for instances, which each get refined separately into instance masks using a convolutional neural network applied to the pixels in the box.
Mask-R-CNN adopts the \emph{mask classification} approach, as each instance mask is classified at once.
Conversely, early semantic segmentation models, e.g.~\cite{long15fully} adopt a \emph{per-pixel classification} approach, where the output is a class prediction for each pixel.

The DETR \cite{carion20DETR} model marked a major step for object detection in the Transformer era.
It replaces the traditional region proposals with learnable object queries, which get mixed with image features in a Transformer decoder to directly output box predictions.
During training, a bipartite matching step determines the optimal matching between groundtruth and predicted objects.
Improvements to DETR include modifications to the attention mechanism \cite{zhu2020deformable}, query construction \cite{liu22dab-detr,zhang22dino}, or improvements to the matching scheme \cite{Zong23hybrid-DETR}.

The DETR architecture inspired a number of related models for segmentation tasks \cite{cheng2021maskformer,m2f,kirillov2023segment_anything,jain2023oneformer,Li23mask_dino,yu22kmax-deeplab}.
The ideas for all these are similar: a series of transformer decoder layers refine a set of mask queries, which get matrix multiplied with image features to generate mask predictions.
This simple design allows the models to be adapted to solve any of the three segmentation tasks with only a minor change to the inference procedure.

\subsection{Diarization models}
Diarization models can be divided into roughly two classes: clustering based systems, and end-to-end models.
Clustering based systems often rely on separate components applied in sequence to output a diarization result.
One example is the VBx system \cite{diez19hmm_vbx,landini22vbx}, which incorporates a voice activity detector (VAD), a speaker embedding extractor, and a clustering algorithm which uses a hidden Markov model for tracking speaker turns.

End-to-End Neural Diarization (EEND) \cite{fujita19eend} formulates diarization as a multi-label classification problem.
The most recognized end-to-end architecture is EEND-EDA \cite{horiguchi20eend-eda}, which is designed to accommodate diarization for an unspecified number of speakers.
However, due to computing the permutation invariant loss by brute force, the original EEND-EDA model was only ever publicly trained to diarize for a maximum of five speakers.
Consequently, it shows poor generalization performance when presented with recordings that contain six or more speakers.

Enhancements to the EEND-EDA model have so far included replacing the self-attention layers of the encoder with the Conformer architecture \cite{liu21conformer,leung21conformer}, improving the prediction head \cite{Wang23EDA-TS-VAD}, and updating the LSTM-based EDA module.
Some studies have attempted to improve the zero vector input to the  LSTM decoder, incorporating the use of an attention mechanism \cite{pan22EEND-spkattn} and summary representations \cite{broughton23summary}.
Nevertheless, the inherent sequential nature of the LSTM architecture results in longer training times and difficulties in handling long-term dependencies due to the vanishing gradient problem.
Non-autoregressive approaches to updating the EDA module have mostly involved swapping the LSTM layers with Transformer variants \cite{samarakoon23EEND-TA,fujita23intermediate,rybicka22end,chen23AED-EEND-EE}.

Other recent works blur the line between clustering and end-to-end models.
One popular method uses EEND on short sliding windows to create local diarization results, which are stitched together by a clustering algorithm \cite{kinoshita21eend-vc,kinoshita21eend-vc-2,plaquet23pyannote31,bredin23pyannote21,horiguchi21eend-gla}.
Others implement EEND with additional input requirements obtained from auxiliary models. 
These may include speaker identity vectors \cite{medennikov20TS-VAD,cheng23SS-TS-VAD} obtained from speaker verification models, or even entire preliminary diarization results \cite{he23ANSD-MA-MSE} computed by another diarization model.

\subsection{Connections and differences}
Comparing the structures of the Maskformer family with the EEND-EDA family reveals striking similarities.
Both models start with a \emph{backbone}, an encoder stack transforming an input image/audio into a latent low-resolution representation.
Next, the models branch out into a query decoder, which generate a number of embeddings, or attractors, corresponding to objects or speakers.
A classification layer determines which embeddings are valid objects/speakers, after which a mask is computed via a matrix product between valid embeddings and the latent representation.

These similarities become obvious when we view diarization as an instance segmentation problem.
The diarization output of a length $T$ recording with $S$ speakers is represented by a $T \times S$ boolean matrix $Y = (y_{t,j})$, where $y_{t,j} = 1$ if speaker $j$ is active at frame $t$.
Correspondingly, a $W \times H$ image with $S$ objects is represented by $S$ boolean masks, organized in a $W \times H \times S$ tensor $Y = (y_{i,j,k})$, where $y_{i,j,k} = 1$ if object $k$ is present at the pixel $(i,j)$.
Hence, speaker diarization is equivalent to finding masks for $S$ object instances in a 1-dimensional, $T$ sized image.

Compared to instance segmentation models, the shift to the audio domain allows for a few simplifications.
The move from 2-dimensional images to 1-dimensional audio enables an increase in resolution by an order of magnitude.
As speech can overlap, we explicitly allow overlapping mask prediction and therefore do not perform any non-maximal suppression (NMS) \cite{Kirillov19panoptic}.
The classification head used in image segmentation models is also simplified for diarization, as we only have two classes (no-speaker/speaker), as opposed to possibly hundreds (no-class, car, tree, person\dots).

\section{EEND-M2F}
\label{sec:model}

\begin{figure}[t]
  \centerline{\includegraphics[width=0.98\linewidth]{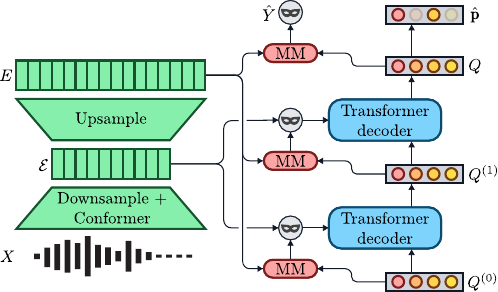}}
  \caption{Overview of EEND-M2F. An encoder backbone, depicted in green, processes the input audio into low and high resolution acoustic features.
  A set of learnable queries are iteratively refined using transformer decoders.
  A mask module (MM, in red) generates speaker-wise masks from queries and acoustic features.
  These masks correspond to diarization predictions, and they are also used to mask the acoustic features in the transformer decoder.
  Finally, an MLP determines which queries correspond to actual speakers.}
  \label{fig:m2f_net}
\end{figure}

We present an overview of the architecture of EEND-M2F, illustrated in \cref{fig:m2f_net}.
Our model takes as input a sequence $X = (x_{i,j})$ of dimensions $T \times D'$, and outputs a speaker-wise boolean mask $\hat Y = (y_{i,j})$ of dimensions $T \times \hat S$, where $\hat S$ is an estimate for the number of speakers.

As mentioned in \cref{sec:intro}, we base our model on Mask2Former \cite{m2f} due to its simplicity and strong performance in a variety of segmentation tasks.
We modify certain design decisions of Mask2Former to better fit the speech domain; each modification is motivated in the corresponding subsection.
Parameter choices and other implementation details can be found in \cref{ssec:implementation} and \cref{tab:hparams}.

\subsection{Feature backbone}
The backbone, depicted in green in \cref{fig:m2f_net}, maps the input sequence into two latent acoustic sequences of different resolutions.
We swap out Mask2Former's image encoders (Swin/ResNet) for one better suited for speech.
For simplicity, we choose a standard Conformer \cite{gulati2020conformer} as our backbone.
First, the input sequence $X$ gets downsampled via convolutional layers to a 1/10th the resolution, after which it passes through $N_\mathrm{conf}$ Conformer layers to produce the low-resolution latent sequence $\ee$, of shape $\frac{T}{10} \times D$.
The full-resolution latent sequence $\EE$ is obtained by a temporal convolutional upsampling of $\ee$ back to the shape $T \times D$.
\begin{align}
\begin{split}
    \ee &= \Conformer(\ConvDown(X)), \\
    \EE &= \ConvUp(\ee).
\end{split}
\end{align}
Note that unlike image segmentation, where objects far away from the camera appear at lower resolution, speech always appears at the same temporal resolution.
Therefore, we do not consider outputting features at additional resolutions with e.g.~a feature pyramid network \cite{lin2017fpn}.

\subsection{Mask module}
The mask module, in red in \cref{fig:m2f_net}, combines the acoustic features with queries to generate speech activity probabilities for each query.
These are then thresholded to create masks for each speaker.

Given $N$ queries $Q$ of shape $N \times D$, we compute set of $N$ diarization logits by simple matrix multiplication
\begin{align}
    \MaskModule(Q, \EE) := \EE \cdot \MLP(Q)^T,
\end{align}
where $\MLP(\cdot)$ denotes a multi-layer perceptron.
Speech activity probabilities are obtained by passing the output through a sigmoid layer, that is
\begin{align}
    \tilde Y = \sigma(\MaskModule(Q, \EE)).
\end{align}

\subsection{Query Module}
In this subsection we explain how the queries $Q$ are constructed.
This module is the equivalent of the EDA module in EEND-EDA; we use the word \emph{query} instead of \emph{attractor} as it better describes the role of $Q$ in the transformer decoder.

For this task we use a stack of $L$ mostly standard Transformer decoder layers, with the only modification being swapping the order of layers so that cross-attention happens before self-attention.
In the cross-attention layers, we mask frames in $\ee$ from each intermediate query using the diarization result obtained from the mask module.
Masking happens \emph{per query}, meaning that each query will see different acoustic information.

More precisely, the query modules take as input a set of queries $Q^{(\ell)}$ of shape $N \times D$, and the latent representation $\ee$ of shape $\frac{T}{10} \times D$, and outputs a set of refined queries $Q^{(\ell + 1)}$, of the same shape as $Q^{(\ell)}$.
First, we compute masked cross-attention between $Q^{(\ell)}$ and $\ee$.
As the mask, we use the intermediate diarization logits $M^{(\ell)} = \Downsample(\MaskModule(Q^{(\ell)}, \EE))$ hard-thresholded at 0 (i.e.~at a probability of 0.5).
Here, $\Downsample(\cdot)$ denotes a simple linear interpolation.
The result then passes through the standard Transformer self-attention and feed-forward layers to output the refined queries.
To summarize
\begin{align}
    \begin{split}
        \label{eq:transformer_decoder}
        Q' &= \LN(\MHA(Q^{(\ell)} + P, \EE, \EE; M^{(\ell)}) + Q^{(\ell)}),\\
        Q'' &= \LN(\MHA(Q' + P, Q' + P, Q') + Q'),\\
        Q^{(\ell + 1)} &= \LN(\FF(Q'') + Q'')
    \end{split}
\end{align}
where $\LN(\cdot)$ denotes LayerNorm \cite{ba2016layernorm}, $\FF$ is a feed-forward layer, and $\MHA(q, k, v; m)$ denotes multi-head attention, with query $q$, key $k$, and value $v$.
The matrix $P$ is a learned, randomly initialized positional encoding for the queries.
The optional input $m$ is a binary mask, which has the effect of adding $-\infty$ to the attention weight $a_{i,j}$ whenever $m_{i,j} = 0$, effectively hiding frame $j$ from query $i$.

The initial queries $Q^{(0)}$ are randomly initialized, learned parameters.
The final set of queries $Q^{(L)}$ are used for inference, that is $Q := Q^{(L)}$.

To explain the motivation behind masked attention, we equate each query with a speaker.
Without masked attention, any given query needs to attend to frames where the corresponding speaker is active, and the model needs to learn to ignore other frames.
For speakers with very little speech in long conversations, the signal from the few frames of speech activity may drown in the noise of all other speech frames.
Masked attention avoids this issue by design: the frames irrelevant to the query are simply masked away, allowing the attention mechanism to focus on relevant acoustic information only.
In this way, model capacity is not wasted in learning to ignore certain frames.
This is in fact similar in spirit to the approach taken by certain TS-VAD based diarization models \cite{he23ANSD-MA-MSE}, which rely on pre-computed diarization results to select frames belonging to each speaker.
In comparison, our method works end-to-end, in a single model, and a single encoder pass.

\subsection{Query classification module}
\label{ssec:class_module}
In effect, the $T \times N$ matrix $\tilde Y$ corresponds to a prediction for $N$ speaker ``proposals'', where $N$ is usually much larger than the ground-truth number of speakers $S$.
We will employ a simple classification layer to decide which columns to cut from $\tilde Y$.
We get a vector $\hat{\vec p}$ of probabilities
\begin{align}
    \hat{\vec p} = \sigma(\Lin(Q)).
\end{align}
During inference, the vector $\hat{\vec p}$ is hard-threhsolded to a fixed value $\theta$.
The final output $\hat Y$ is constructed from $\tilde Y$ by keeping columns $i$ for which $p_i > \theta$.
The estimated speaker count $\hat S$ is thus the number of entries of $\hat{\vec p}$ that are greater than $\theta$.
Essentially, one can think of the $N$ queries as mask proposals, which are either accepted or rejected based on the result of the classification layer.
When training, the output of the Hungarian matching step (\cref{ssec:matching}) determines which of the speaker proposals correspond actual speakers.

In EEND-EDA predicted attractors are classified as either being a speaker or a non-speaker, and the model is trained so that the first $S$ attractors are classified as speakers and the $S+1$ as a non-speaker.
This means that some model capacity needs to be used to organize and reorder speaker attractors into the first $S$ positions of the attractor sequence.
While self-attention makes this permutation operation fairly simple to learn in transformers, we avoid it altogether by simply ignoring the output ordering of the queries.
This is precisely the classification strategy of Mask2Former.

\subsection{Hungarian matching}
\label{ssec:matching}

\begin{figure}[t]
    \centerline{\includegraphics[width=0.8\linewidth]{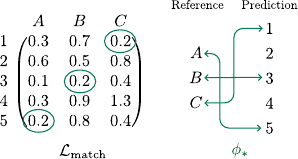}}
    \caption{Example of the Hungarian matching procedure with $N=5$ queries (1,2,3,4,5) and $S=3$ speakers (A,B,C).
    A matching is an assignment of each of the $S$ speakers to a unique query.
    The cost matrix consists of pairwise matching costs, and the optimal matching $\phi^*$ is the one minimizing the sum of $S$ costs, where at most one is chosen from each row and each column.
    }
    \label{fig:matching}
  \end{figure}
  
Both end-to-end diarization and image segmentation are inherently \emph{set prediction} problems: any permutation of a model output or reference labels should be deemed equally valid.
During training, we make the loss permutation invariant by assigning each of the $S$ ground-truth speakers to a unique prediction among the $N$ possibilities.
The optimal matching $\phi^*$ is the element of the set $\mathcal{P}^N_S$ of $S$-permutations of $N$ that minimizes the matching criterion $\mathcal{L}_\mathrm{match}(\hat Y_\phi, Y)$, where $\hat Y_\phi$ is the ordered subset of columns of $\hat Y$ described by $\phi \in \mathcal{P}^N_S$.
A simple example is shown in \cref{fig:matching}, with $N=5$ and $S = 3$.

The matching cost function $\mathcal{L}_\mathrm{match}$ consists of three weighted components
\begin{align}
    \begin{split}
    \mathcal{L}_\mathrm{match} = &\lambda_\mathrm{dia}\mathcal{L}_\mathrm{match,dia} \\&+ \lambda_\mathrm{dice}\mathcal{L}_\mathrm{match,dice} + \lambda_\mathrm{cls}\mathcal{L}_\mathrm{match,cls},
    \end{split}
\end{align}
similar to the training loss function described in \cref{ssec:losses}.
The diarization matching cost $\mathcal{L}_\mathrm{match,dia}$ is the sum of binary cross entropies, averaged over time.
The dice matching cost $\mathcal{L}_\mathrm{match,dice}$ is the sum of dice losses (c.f.~\cref{sssec:dice}).
Finally, we add a classification penalty to discourage the selection of predictions classified as ``non-speaker''
\begin{align}
    \mathcal{L}_\mathrm{match,cls} = -\sum_{i=1}^S \hat{p}_{\phi(i)}.
\end{align}
We choose probabilities instead of logits to keep the scale comparable with the other components.
We observe that including a classification penalty to the matching loss is essential to avoid duplicated predictions during inference.
A possible explanation could be that it introduces a stronger feedback loop ensuring a consistent query-to-speaker matching across training epochs.

Note that the cost function is a sum of pairwise matching costs (between reference speaker and predicted output).
Therefore, an efficient implementation would first construct the $N \times S$ cost matrix, and then use the Hungarian algorithm \cite{kuhn55hungarian} to find the optimal matching $\phi^* \in \mathcal{P}_S^N$.
An efficient implementation of the Hungarian algorithm, \texttt{scipy.optimize.linear\_sum\_assignment}, can be found in SciPy \cite{scipy}.

\subsection{Losses}
\label{ssec:losses}
The loss function optimized during training is a linear combination of three components: the diarization loss $\LL_\mathrm{dia}$, the Dice loss $\LL_\mathrm{dice}$, and the classification loss $\LL_\mathrm{cls}$
\begin{align}
    \LL =&\lambda_\mathrm{dia}\LL_\mathrm{dia} + \lambda_\mathrm{dice}\LL_\mathrm{dice} + \lambda_\mathrm{cls}\LL_\mathrm{cls}.
\end{align}
These losses are computed based on the optimal permutation $\phi^*$ obtained in the Hungarian matching step.
\subsubsection{Diarization loss}
Diarization loss is computed using binary cross entropy, as is standard for most diarization systems and Mask2Former.
We only keep predictions corresponding to queries matched to actual speakers in the Hungarian matching step.
Therefore, for each recording, the loss is of the form
\begin{align}
    \LL_\mathrm{dia} = \sum_{i = 1}^S \sum_{t = 1}^T H(\hat{Y}_{t,\phi^*(i)}, Y_{t,i}),
\end{align}
where $H(p,q) = -q\log p - (1-q)\log(1-p)$.
For example, in the situation of \cref{fig:matching}, the diarization loss would be the sum of the binary cross entropies between speaker A and prediction 5, speaker B and prediction 3, and speaker C and prediction 1.
In a batch, we compute the weighted average of the above quantities, using the product $TS$ as the weight for each recording.
\subsubsection{Dice loss}
\label{sssec:dice}
We add the dice loss \cite{milletari16dice} as a novel addition to diarization systems.
As with the diarization loss, only matched speakers contribute to the dice loss.
For each recording, the loss is
\begin{align}
    \LL_\mathrm{dice} = 1 - \frac{1}{S}\sum_{i = 1}^S \frac{2\sum_{t=1}^T \hat{Y}_{t,\phi^*(i)}\cdot Y_{t,i}}{\sum_{t=1}^T \hat{Y}_{t,\phi^*(i)} + \sum_{t=1}^T Y_{t,i}}.
    \label{eq:dice_definition}
\end{align}
In a batch, the dice losses for each recording are averaged, with weighting given by $S$ for each recording.

One of the benefits of dice loss is its scale invariance: every speaker contributes to the loss equally, regardless of the amount of their speech in the conversation.
Dice loss also directly targets missed speech, as only frames with reference speech ($Y_{t,i} = 1$) contribute to the loss.

\subsubsection{Query classification}
Recall that the model outputs a vector $\hat{\vec p}$ of probabilities, indicating which of the $N$ queries is likely to correspond to an actual speaker.
During training, the reference probability vector $\vec p$ is given by the result of the Hungarian matching: $p_i = 1$ if the $i$th query is matched to a reference speaker (i.e.~$\phi^*(j) = i$ for some $j = 1,\dotsc,S$), otherwise it is set to $0$.
For example in \cref{fig:matching}, we would have $\vec p = (1,0,1,0,1)$.
The classification loss is the binary cross entropy between $\vec p$ and $\hat{\vec p}$.
Since $S$ is usually much smaller than $N$, to account for the class imbalance, we weigh the negative class ($p_i = 0$) down by a factor of $0.2$.
\begin{align}
    \LL_\mathrm{cls} = \sum_{i=1}^N \left[ 1.0 \cdot p_i + 0.2 \cdot (1-p_i) \right] \cdot H(\hat{p}_i, p_i)
\end{align}
In a training batch, the classification loss $\LL_\mathrm{cls}$ is the weighted average of all binary cross entropies.

\subsubsection{Deep supervision}
Similar to recent diarization papers \cite{landini23diaper,fujita23intermediate}, we use deep supervision during model training.
Following Mask2Former, we only supervise intermediate outputs from the query module, using every set of queries $Q^{(0)},Q^{(1)},\dotsc,Q^{(L)}$.
This means that for each recording, the Hungarian matching and loss computation is performed $L+1$ times.
Deep supervision is essential when using masked attention, as it ensures that intermediate masks correspond to reasonable diarization predictions.

\section{Experiments}
\label{sec:experiments}
\subsection{Datasets}
\label{ssec:datasets}
Pretraining data is generated from the LibriSpeech corpus \cite{librispeech} using the standard diarization mixture simulation algorithm \citep[Alg.~1]{fujita19eend}.
We generate a total of $400\,000$ mixtures containing 1, 2, 3, or 4 speakers, with average silence interval $\beta = 2,2,5,9$ respectively.
Contrary to previous diarization works, we found no benefit in starting pretraining with 2 speaker simulations only and therefore we pretrain directly on the entire simulated dataset.

For finetuning, we use a selection of several publicly available diarization datasets: Aishell-4 \cite{aishell4}, AliMeeting \cite{alimeeting}, AMI-SDM (single distant microphone) and AMI-Mix (hedaset mixtures)\cite{AMI}\footnote{full-corpus ASR partitions, see \url{https://github.com/BUTSpeechFIT/AMI-diarization-setup}}, CALLHOME \cite{callhome}\footnote{standard Kaldi split \url{https://github.com/kaldi-asr/kaldi/tree/master/egs/callhome_diarization/v2}}, DIHARD-III (``full'' partition), RAMC \cite{ramc}, and VoxConverse (v.~0.3) \cite{voxconverse}.
When provided, we use the official training/validation/testing splits provided by the dataset.
If only training and evaluation data is provided, we use the evaluation data as a test set, and split the training set into training and validation subsets using a 80\%/20\% split.
In case of multi-channel audio, we mix all channels into one.
For the AliMeeting test set, we refer to the far-field 8 microphone array recordings, used in the M2MeT challenge, as ``AliMeeting-far'', whereas ``AliMeeting-near'' refers to the headset microphone recordings.

\subsection{Implementation details}
\label{ssec:implementation}
\textbf{Backbone}
The model takes as input a sequence of 23-dimensional log-Mel features, extracted every 10 ms with a window size of 25 ms.
The input get downsampled into $D = 256$ dimensional features by a depthwise separable convolution, with hop size 10 and kernel size 15, followed by layer normalization and dropout with probability 0.1.

Downsampled features are fed through $N_\text{conf} = 6$ Conformer layers \cite{gulati2020conformer}, with 4 attention heads and convolution kernel of size 49.
We use layer normalization instead of the usual batch normalization to avoid contaminating batch statistics with padding.

Upsampling is performed by applying a 1-dimensional transposed convolution, a layer normalization, and a GELU activation.
The upsampling block consists of two such sequences, with kernel sizes 3, 5 and strides 2, 5.

\textbf{Mask \& query modules}
Before matrix multiplication, queries are passed through an MLP layer, consisting of two hidden layers with ReLU activations.
The classification layer is a single linear layer.

\textbf{Query module}
The query module consists of $L = 6$ modified Transformer decoder layers (with self-attention and cross-attention swapped).
We use multihead attention with 4 heads, feedforward dimension 1024 and no dropout.
To mask the low resolution features in the cross-attention layer, we downsample the high resolution intermediate mask by linearly interpolating the logits, and thresholding at a probability of 0.5.
We initialize $N=50$ learnable queries, each of dimension 256, along with a set of $50 \times 256$ learnable positional encodings.

\textbf{Training and inference}
After pretraining, we only keep the weights of the backbone for finetuning, and randomize other parameters.
This allows for fast model retraining when only diarization head parameters are changed.
After each finetune, the 10 best checkpoints in terms of validation DER are averaged to create the final model weights.
We start by finetuning on an aggregated dataset consisting of all training data, resulting in a single general model labeled EEND-M2F in \cref{tab:main}.
This model is further finetuned on single datasets, resulting in the row EEND-M2F + FT in \cref{tab:main}.

Numerical hyperparameters can be found in \cref{tab:hparams}.
We use the AdamW optimizer \cite{loshchilov18adamw}, but we observe that removing weight decay from the standard Mask2Former training parameters improves performance slightly.
We also observe a very fast overfitting of the classification loss, which is mitigated by label smoothing \cite{szegedy16label_smoothing}.

During inference, we predict all $N=50$ probabilities $p_i$ and diarization results in parallel, and we discard all predictions where $p_i < 0.8 =: \theta$.
As a consequence, EEND-M2F can predict at most 50 speakers at a time, which is well within the limits of current datasets.
Final model outputs are thresholded at 0.5.

Both training and inference are conducted using bfloat16 precision.
Under these settings, our model has a total of $16.3$ million parameters.
Pretraining takes around 90 hours on 4 Nvidia A6000 GPUs, and finetuning takes about 12 hours on one A6000.
The approximately 158 hour long test set takes around 100 seconds to process, making the model about 5,700 times faster than real time.

\begin{table}
    \centering
    \footnotesize
    \begin{tabular}{lccc}
        \toprule
        & & \multicolumn{2}{c}{Finetune} \\
        \cmidrule{3-4}
         & Pretrain & EEND-M2F & +FT \\
        \midrule
        Batch size ($\times$ GPUs) & $128$ ($\times 4$) & $32$ ($\times 1$) & $8$ ($\times 1$) \\
        Utterance length (s) & $50$ & $300$ & $600$ \\
        Max learning rate & $1\cdot 10^{-4}$ & $5\cdot 10^{-5}$ & $5\cdot 10^{-6}$\\
        Steps & $500\,000$ & $50\,000$ & $10\,000$\\
        Learning rate schedule & 1-cycle & \multicolumn{2}{c}{constant} \\
        Backbone dropout & $0.1$ & \multicolumn{2}{c}{$0.1$} \\
        Query module dropout & $0.0$ & \multicolumn{2}{c}{$0.0$} \\
        Weight decay & $0.0$ & \multicolumn{2}{c}{$0.0$} \\
        Label smoothing & $0.0$ & \multicolumn{2}{c}{$0.1$} \\
        $\lambda_\mathrm{dia}, \lambda_\mathrm{dice}, \lambda_\mathrm{cls}$ & $5,5,2$ & \multicolumn{2}{c}{$5,5,2$}\\
        \bottomrule
        \end{tabular}
    \caption{Detailed hyperparameters. After pretraining, we finetune on a concatenation of all datasets using settings in the ``EEND-M2F'' column. Finally, we further finetune on each dataset separately, using settings in the ``+FT'' column.}
    \label{tab:hparams}
\end{table}

\subsection{Main results}
Our main results, along with the current state-of-the-art, can be seen in \cref{tab:main}.
Results are given in the most realistic setting, scoring all speech including overlaps, with no oracle voice activity detection, no oracle speaker counting, and without any no-scoring collar.
It is noteworthy to mention that every previous state-of-the-art result come from models specifically finetuned for a single dataset.
EEND-M2F attains a new state-of-the-art for AliMeeting, DIHARD-III and RAMC, without any dataset-specific finetuning.
Further finetuning allows EEND-M2F to achieve state-of-the-art performance on AMI-Mix and AMI-SDM, with further improvements on other datasets.

One of the shortcomings of EEND-M2F seems to be audio with a high number of speakers.
Such datasets include Aishell-4, with audio ranging from 5 to 7 speakers, and VoxConverse, with recordings containing up to 21 speakers.
In these scenarios, classical clustering based systems are still superior.
Theoretically, we expect the Mask2Former architecture to perform well even when the number of speakers is high, as the same architecture applied to image segmentation routinely handles tens of objects at a time. 
We leave this investigation for future work.

Another dataset in which EEND-M2F is far from state-of-the-art is CALLHOME.
A possible explanation for poor performance may be channel mismatch.
Our model is pretrained and finetuned on a wide variety of 16kHz audio, with the only exception being CALLHOME with its 8kHz telephony audio.
Since CALLHOME is such a small dataset, we suspect that the model may be poorly calibrated to adapt to such a domain.
In fact, the standard training recipe for CALLHOME, as described in \cite{fujita19eend}, pretrains on simulations drawn from 8kHz telephony datasets, which dramatically diminishes domain mismatch upon finetuning.

\subsection{Ablations}
Results in this section are reported on the concatenation of all test sets described in \cref{ssec:datasets}.
As before, our evaluation metric considers overlapped speech, with no collar and no oracle segmentation.

\subsubsection{Deep supervision \& masked attention}
In \cref{tab:deep_sup}, we study the effect of masked attention (MA) and deep supervision (DS).
We see that both features steadily improve performance overall.

We note that Mask2Former requires deep supervision to work optimally, as intermediate predictions are fed back into the model and need to correspond to reasonable diarization results.
In the absence of deep supervision, gradient updates will never directly improve intermediate predictions due to the hard threshold.
This explains the dramatic increase in error rate when masked attention is used without deep supervision.
Once masked attention is removed however, deep supervision is no longer strictly necessary.
\begin{table}
    \centering
    \begin{tabular}{llrrrr}  
        \toprule
        MA & DS & MS & FA & SE & DER \\
        \midrule
        \xmark & \xmark & 7.43 & 4.75 & 3.81 & 15.99\\
        \cmark & \xmark & 16.24 & 4.36 & 5.76 & 26.36\\
        \xmark & \cmark & 6.75 & 4.59 & 4.29 & 15.63\\
        \cmark & \cmark & 6.38 & 4.60 & 4.19 & 15.17\\
        \bottomrule
        \end{tabular}
    \caption{Diarization error rates (DER, \%) and its summands (missed speech MS, false alarms FA, speaker errors SE) when ablating masked attention (MA) and deep supervision (DS).
    }
    \label{tab:deep_sup}
\end{table}

\subsubsection{Loss components}
As dice loss is a new addition to diarization models, in \cref{tab:losses} we analyze its impact on DER.
For certain datasets in image segmentation, dice loss alone can outperform cross-entropy losses \cite{milletari16dice}.
We observe however that dice loss is clearly not as vital for diarization models as it is for image segmentation models.
There is in fact a slight degradation in DER when adding dice loss, which comes from lowering the false alarm rate in exchange for a slight increase in speaker errors and missed speech.
Since in most practical applications missed speech is more impactful than false alarms, we accept the minor hit in DER for the improved missed speech performance enabled by dice loss.
The reason why dice loss alone is not enough is evident from \cref{eq:dice_definition}: dice loss only supervises frames containing reference speech, making a more global binary cross entropy loss necessary.

\begin{table}
    \centering
    \begin{tabular}{ccrrrr}  
        \toprule
        $\mathcal{L}_\mathrm{dia}$ & $\mathcal{L}_\mathrm{dice}$ & MS & FA & SE & DER \\
        \midrule
        \cmark & \xmark & 6.80 & 4.08 & 4.25 & 15.14\\
        \xmark & \cmark & 14.19 & 11.76 & 5.46 & 31.42\\
        \cmark & \cmark & 6.38 & 4.60 & 4.19 & 15.17\\
        \bottomrule
        \end{tabular}
    \caption{Effect of loss components on DER and its summands.}
    \label{tab:losses}
\end{table}

\subsubsection{Number of queries}
\emph{A priori}, the only requirement for the number of queries is that there are at least as many of them as the maximum number of speakers in a single recording, which in our case is 21 speakers in a single VoxConverse recording.
A higher number of queries may be beneficial, as they can act as a ``memory'' and lead to more flexible models, even when the extra queries are never used.
On the other hand, a lower number may lead to more focused training, as there are fewer possibilities to choose from in the Hungarian matching step.
\Cref{tab:queries} tallies DERs of a set of models trained with different number of queries.
While differences are small, the model using 50 queries seems to have struck a good balance in the number of queries.
\begin{table}
    \centering
    \begin{tabular}{lrrrr}  
        \toprule
        \# queries ($N$) & MS & FA & SE & DER \\
        \midrule
        25 & 6.38 & 4.65 & 4.20 & 15.23 \\
        50 & 6.38 & 4.60 & 4.19 & 15.17 \\
        75 & 6.30 & 4.91 & 4.22 & 15.43 \\
        \bottomrule
        \end{tabular}
    \caption{Diarization error rates with different numbers of queries.}
    \label{tab:queries}
\end{table}

\subsubsection{Number of layers}
In \cref{tab:layer} we experiment with different layer configurations.
Unlike most Conformer based EEND models \cite{liu21conformer,leung21conformer,chen23AED-EEND-EE,samarakoon23EEND-TA}, we observe an improvement by increasing the number of Conformer layers from 4 to at least 6 in the backbone.
Once backbone size is fixed, varying the number of masked-attention transformer decoder layers in the mask module has minor effects.
\begin{table}
    \centering
    \begin{tabular}{ccccc}
        \toprule
        & & \multicolumn{3}{c}{Mask module} \\
        \cmidrule(lr){3-5}
        & & 4 & 6 & 8 \\
        \midrule
        \multirow{3}{*}{\vertical{Backbone}} & 4 & 16.41 & 16.43 & 16.48\\
        \addlinespace[0.2em]
        & 6 & 15.20 & \textbf{15.17} & 15.20\\
        \addlinespace[0.2em]
        & 8 & 15.59 & 15.22 & 15.66\\
        \bottomrule
    \end{tabular}
    \caption{Diarization error rates (\%) with different numbers of layers.
    Backbone refers to the number of Conformer (encoder) layers in the backbone.
    Mask module refers to the number of Masked-attention transformer (decoder) layers.}
    \label{tab:layer}
\end{table}

\section{Discussion}
\label{sec:discussion}
In this paper we present EEND-M2F, a novel end-to-end neural diarization model.
Its simple design, which borrows from advances in image segmentation, sets a strong baseline for future models in terms of generalized performance, speed and parameter efficiency.
EEND-M2F achieves state-of-the-art results in several datasets, such as AliMeeting, AMI, CALLHOME, DIHARD-III and RAMC.

The connection between image segmentation and diarization opens the doors for potential future research.
Recent work in object detection has shown that \emph{non-parametric} queries, e.g. obtained by sampling the input image, can outperform globally initialized queries \cite{liu22dab-detr}.
In diarization, TS-VAD based methods \cite{medennikov20TS-VAD, Wang23EDA-TS-VAD, cheng23SS-TS-VAD} operate on a similar principle, but since they require an initial diarization result as an input they cannot be considered end-to-end.
Following the success of the Segment Anything Model \cite{kirillov2023segment_anything}, one could also envisage a more foundational model in speech processing, unifying several segmentation tasks such as diarization, target speech extraction and source separation using a variety of prompt types.
Finally, this study is also a testament to the versatility of Maskformer-style architectures for segmentation tasks.
In addition to diarization, Maskformer variants are now used for state-of-the-art results in instance segmentation, semantic segmentation, panoptic segmentation \cite{m2f}, video instance segmentation \cite{m2f_video}, 3D object segmentation \cite{Schult23mask3d}.

\bibliography{refs}
\bibliographystyle{icml2023}

\newpage
\appendix
\onecolumn
\section{Efficient Diarization Error Rate computations.}
\label{app:der}
\definecolor{bg}{rgb}{0.95,0.95,0.95}
\begin{listing}
\begin{minted}[linenos,mathescape,style=default,bgcolor=bg]{python3}
import torch
from scipy.optimize import linear_sum_assignment

def compute_DER(sys, ref):
    # sys: $T \times \hat{S}$ boolean torch.Tensor
    # ref: $T \times S$ boolean torch.Tensor

    nsys = sys.sum(1)
    nref = ref.sum(1)

    correct = torch.logical_and(sys[:, :, None], ref[:, None, :])
    matching = linear_sum_assignment(correct, maximize = True)
    ncor = correct[matching].sum()

    ms = torch.clamp(nref - nsys, 0).sum()
    fa = torch.clamp(nsys - nref, 0).sum()
    se = torch.min(nsys, nref).sum() - ncor
    de = torch.max(nsys, nref).sum() - ncor # == ms + fa + se

    z = nref.sum()

    return ms/z, fa/z, se/z, de/z
\end{minted}
\caption{Efficient PyTorch implementation of DER computation.}
\label{code:DER}
\end{listing}
The standard evaluation metric for diarization systems is the diarization error rate (DER), introduced by the National Institute of Standards and Technology \cite{nist09der}.
Though the DER metric is extremely simple, most publicly available DER computation methods often span hundreds or even thousands of lines of code and require several dependencies.
Many versions, including \texttt{md-eval.pl}\footnote{https://github.com/nryant/dscore} and spyder\footnote{https://github.com/desh2608/spyder}, operate on speech segments instead of speaker masks, making reference and system speaker matching slow due to sequential processing.

In \cref{code:DER}, we present our version of a DER computation function.
It is short and concise, only taking around 10 lines of Python code.
Since our version is directly compatible with model outputs, the entire computation can be performed in parallel on a GPU, with no virtually memory overhead during training.
The only CPU operation required is an execution of the Hungarian algorithm on a small matrix ($S \times \hat S$), which in practice is extremely fast.

While similar versions of \cref{code:DER} exist in publicly available repositories, they often rely on proxy objectives and explicit brute-force searches for reference and system speaker matching.
The closest implementation is the one in PyAnnote \footnote{\url{https://github.com/pyannote/pyannote-audio}, as of commit 80634c9}, however it uses mean square error as a proxy to match system and reference speakers.
Other versions use the training criterion (binary cross entropy) to determine the correct matching.
These heuristics may result in suboptimal DER numbers, especially in cases where the model is overconfident about its predictions.
Our version on the other hand comes with a guarantee of optimality.

\begin{proposition}
    The algorithm in \cref{code:DER} computes the optimal diarization error rate.
\end{proposition}
\begin{proof}
    The NIST specification \cite{nist09der} defines diarization error rate as
    \begin{align}
        \text{DER} = \dfrac{\sum_{\substack{\text{all}\\\text{segs}}} \dur(\seg) \left(\max(N_\mathrm{ref}(\seg), N_\mathrm{sys}(\seg)) - N_\mathrm{cor}(\seg) \right)}{\sum_{\substack{\text{all}\\\text{segs}}} \dur(\seg) N_\mathrm{ref}(\seg)},
    \end{align}
    where $N_\mathrm{ref}(\seg), N_\mathrm{sys}(\seg)$ are correspondingly reference and system number of speakers in a segment, and $N_\mathrm{cor}(\seg)$ is the number of reference speakers in a segment correctly matched to a system speaker.
    Clearly we may subdivide any segment whose both system and reference speaker are constant.
    If we define each frame to be a segment, we may equivalently write
    \begin{align}
        \label{eq:der_formula}
        \text{DER} = \dfrac{\sum_{t=1}^T \max(N_\mathrm{ref}(t), N_\mathrm{sys}(t)) - N_\mathrm{cor}(t)}{\sum_{t=1}^T N_\mathrm{ref}(t)}.
    \end{align}
    Reusing the notation in lines 5 and 6, the quantities $N_\mathrm{ref}(t)$, $N_\mathrm{sys}(t)$ are \mintinline{python3}|nref[t]|, \mintinline{python3}|nsys[t]| respectively.
    The denominator is precisely the quantity $\mathtt{z}$ computed in line 20.
    
    Note that since neither $N_\mathrm{ref}(t)$ nor $N_\mathrm{sys}(t)$ depend on speaker ordering, the optimal DER is obtained by finding the ordering that maximizes $\sum_{t=1}^T N_\mathrm{cor}(t)$.
    Without loss of generality, let $\hat S \geq S$.
    Since unmatched speakers (when $\hat S$ is strictly greater than $S$) don't contribute to $N_\mathrm{cor}$, it suffices to find the $S$-permutation of $\hat S$ speakers $\phi^*$ that maximizes
    \begin{align}
        \sum_{t=1}^T N_\mathrm{cor}(t) = \sum_{i=1}^{S} \sum_{t=1}^t \mathbb{1}_{\{\mathtt{ref}[t,i] \wedge \mathtt{sys}[t, \phi(i)]\}}
    \end{align}
    among all $\phi \in \mathcal{S}^{\hat S}_S$.
    This is a linear sum assignment problem, that can be solved by constructing the reward matrix $\mathtt{correct[j,i]} = \sum_{t=1}^t \mathbb{1}_{\{\mathtt{ref}[t,i] \wedge \mathtt{sys}[t, j]\}}$.
    Line 13 computes $\mathtt{ncor} = \sum_{t=1}^T N_\mathrm{cor}(t)$ for the optimal permutation, and line 18 computes the numerator of \cref{eq:der_formula}.
    It is now clear that the returned quantity $\mathtt{de / z}$ is the optimal DER.

    A similar argument works for missed speech (MS), false alarms (FA) and speaker errors (SE)
    \begin{gather*}
        \text{MS} = \dfrac{\sum_{t=1}^T \max(N_\mathrm{ref}(t) - N_\mathrm{sys}(t), 0)}{\sum_{t=1}^T N_\mathrm{ref}(t)}, \\
        \text{FA} = \dfrac{\sum_{t=1}^T \max(N_\mathrm{sys}(t) - N_\mathrm{ref}(t), 0)}{\sum_{t=1}^T N_\mathrm{ref}(t)}, \\
        \text{SE} = \dfrac{\sum_{t=1}^T \min(N_\mathrm{ref}(t), N_\mathrm{sys}(t)) - N_\mathrm{cor}(t)}{\sum_{t=1}^T N_\mathrm{ref}(t)}.
    \end{gather*}
    The SE metric is optimal for the same reasons as DER.
    MS and FA metrics don't depend on matching.
\end{proof}

From this formulation we can also readily implement no-scoring collars.
While the NIST specification doesn't explicitly describe how collars are applied, a deep dive into \texttt{md-eval.pl} reveals that a collar of length $t$ is equivalent to ignoring an interval of length $2t$ centered at every start and end of segments in the reference diarization labels.
Translating to the notation in \cref{code:DER}, a collar of $0.25$ seconds would correspond to setting \mintinline{python3}|sys[t-25 : t+25, :]| and \mintinline{python3}|ref[t-25 : t+25, :]| to \mintinline{python3}|False| for all \mintinline{python3}|t| starting or ending a sequence of consecutive values of \mintinline{python3}|True| in \mintinline{python3}|ref[:, i]| for any reference speaker \mintinline{python3}|i|.

\end{document}